\documentclass{llncs}

\usepackage[utf8]{inputenc}
\usepackage{latexsym}
\usepackage{amssymb}

\usepackage{mathrsfs}

\usepackage{mathtools}
\usepackage{bussproofs}
\usepackage[colorlinks]{hyperref}
\usepackage[colorinlistoftodos, textwidth=4cm, shadow]{todonotes}

\newcommand{\mysec}[1]{Section~#1}
\newcommand{\mydef}[1]{Definition~#1}

\newcommand{\myfig}[1]{Figure~#1}
\newcommand{\mythm}[1]{Theorem~#1}

\EnableBpAbbreviations
\newcommand{\UICm}[1]{\UIC{$#1$}}
\newcommand{\AXCm}[1]{\AXC{$#1$}}
\newcommand{\BICm}[1]{\BIC{$#1$}}

\def\land{\wedge}
\def\lor{\vee}
\def\limp{\Rightarrow}
\def\fa{\forall}

\def\ex{\exists}

\def\gen{\lrcorner}
\newcommand{\gpos}[1]{{{#1}^p}} 
\newcommand{\gneg}[1]{{{#1}^n}} 
\def\weak{\mbox{weak}}
\def\contr{\mbox{contr}}

\title{Polarizing Double-Negation Translations}
\author{ M\'elanie Boudard\inst{1} \and Olivier Hermant\inst{2}}
\institute{PRiSM, Univ. de Versailles-St-Quentin-en-Yvelines
	\\CNRS, France
	\\ \email{melanie.boudard@prism.uvsq.fr}
    \and CRI, MINES ParisTech 
    \\ \email{olivier.hermant@mines-paristech.fr}
    }

\begin{document}

\maketitle

\begin{abstract}
Double-negation translations are used to encode and decode classical proofs 
in intuitionistic logic. We show that, in the cut-free fragment, we can 
simplify the translations and introduce fewer negations. To achieve this, 
we consider the polarization of the formul\ae{} and adapt those translation 
to the different connectives and quantifiers. We show that the embedding 
results still hold, using a customized version of the focused classical sequent 
calculus. We also prove the latter equivalent to more usual versions of the sequent calculus. This polarization process allows lighter embeddings, and sheds some light on the  
relationship between intuitionistic and classical connectives.\\

\end{abstract}

\section{Introduction}

The relationship between different formal systems is a longstanding
field of studies, and involves for instance conservativity, relative 
consistency or independence problems \cite{PCoh66}. 
As for deductive systems, the natural question is to find a conservative 
encoding of formul\ae{}. By conservative, we mean an encoding of formul\ae{} 
such that a formula is provable in the first system if and only if its 
encoding is provable in the second system. This work was pioneered by Kolmogorov \cite{AKol25}, G\"odel \cite{KGod33} and Gentzen \cite{GGen36} for classical and intuitionistic logics. There exist several classes of sequents that are known to be classically provable if and only if they are intuitionistically provable \cite{Schwichtenberg2013}.

In this paper, we refine those translations by removing a large number of unnecessary negations. Instead of focusing on invariant classes as in \cite{Schwichtenberg2013}, we consider a translation on all the formul\ae{}. A common point with this work, however, is the use of syntactic transformations. The proof systems we consider are the cut-free  intuitionistic and classical sequent calculi \cite{HSchATro96}. This allows two remarks:

\begin{itemize}
\item the left rules of both calculi are identical; therefore it seems 
natural to translate them by themselves, when possible.
\item In the absence of the cut rule, a formula is never active in different sides (both as an hypothesis and as a conclusion) of the turnstyle, having therefore a well-defined {\em polarity}. This last fact holds for all the rules except the axiom rule, which is easily dealt with, by an $\eta$-expansion-like argument, i.e. decomposing the formula by structural rules until we get axioms between atomic formula only.

\end{itemize}

In summary, we can avoid the introduction of negations on formul\ae{} 
belonging to the ``left'' (or hypothesis) side of sequents. We also 
introduce further refinements, inspired by those of \cite{KGod33,GGen36},
to remove even more negations in the translation, based on the
observation that some right-rules are also identical in
the classical and intuitionistic calculi. To show conservativity by 
syntactic means without the cut rule, we need to impose a focusing discipline on the right-hand side of the classical sequent calculus, forced by the single-formula 
condition on the right-hand side of an intuitionistic sequent. We dedicate \mysec{\ref{sec:foc}} to the study of a customized focused sequent calculus.

The price to pay of an asymmetric translation is that the result misses some modularity since we dismiss the cut rule: given a proof of a $A$ and a proof of $A \limp B$, we cannot combine them with a cut rule. Both translations of $A$ are not the same and so the translations of the proofs do not compose {\em directly}. See also the discussion in \mysec{\ref{sec:conc}}.

The paper is organized as follows. In \mysec{\ref{sec:pre}}, we give a brief 
overview of the background material, in particular the negative translations.
In \mysec{\ref{sec:Kpol}}, we introduce a first polarized refinement of 
Kolmogorov negative translation, while \mysec{\ref{sec:foc}} discusses 
the properties of the focused sequent calculus that we need in 
\mysec{\ref{sec:GGpol}} to show that the polarized refinement of 
Gentzen-G\"odel negative translation still has the same properties than 
the other translations. \mysec{\ref{sec:conc}} concludes the paper.

\section{Prerequisites} \label{sec:pre}

Here, we briefly recall the syntax of first-order logic, sequent calculus and the already known double-negation translations.

\subsection{First-Order Logic}

We assume that the reader is familiar with one-sorted first-order logic
\cite{HSchATro96}: terms are composed of variables and function symbols applied to terms along their arities, and formul\ae{} are either predicate symbols applied to terms along their arities or composed ones with the help of the conjunction ($\land$), disjunction ($\lor$), implication ($\limp$), negation ($\neg$) connectives and the universal ($\fa$) and existential ($\ex$) quantifiers.

To shorten the statement of our results and their proofs, we also define an operator that removes the trailing negation of formul\ae{}, if any, and otherwise adds it.

\begin{definition}[antinegation] \label{def:antinegation}
Let $A$ be a formula, we let $\gen A$ be:
\begin{itemize}
\item $B$ if $A$ is equal to $\neg B$
\item $\neg A$ otherwise.
\end{itemize}
\end{definition}

Note that $\gen$ is not a connective, it is an operator, similar to Boolean complement in that $\gen\neg$ is the identity. In particular it has no associated rule in the sequent calculus. For instance $\gen P(a)$ {\em is the same as} $\neg P(a)$ while $\gen \neg (A \land B)$ {\em is the same as} $ (A \land B)$.

\subsection{Sequent Calculi}

Since they will be discussed in details in the next sections, we
explicitly give the details of the classical and intuitionistic sequent
calculi. A sequent is a pair of two multisets of formul\ae{},
denoted $\Gamma \vdash \Delta$. The comma serves as a shorthand for
multi-set union and $\Gamma, A$ is an overloaded notation for $\Gamma, \{ A \}$.\\

The classical sequent calculus is presented in
\myfig{\ref{fig:clas}}. The formula that is decomposed is called the
      {\em active} formula. The intuitionistic sequent calculus
differs from the classical in the restriction imposed to the
right-hand sides of sequents: it must be either empty, or reduced to
one formula. Consequently, the following rules are modified: contr${}_R$
disappears; in the first premiss of the $\limp_L$ rule and the axiom rule, $\Delta$ is empty; finally, the $\lor_R$ rule splits to account for the choice of keeping $A$ or $B$. For clarity, the intuitionistic sequent calculus is presented in \myfig{\ref{fig:intuit}}.

Note that, as announced, we do not consider the cut rule to be part
of the calculus; so we reason in cut-free calculi.

\begin{figure}
  \framebox[12.2cm][c]
{\parbox{12.2cm}
{\scriptsize
  \begin{center}
  \begin{tabular}{c@{\hspace{1cm}}c}
\multicolumn{2}{c}{   \AXCm{~} \RL{ax}
    \UICm{\Gamma, A \vdash A, \Delta} \DP }
\\\\

    \AXCm{\Gamma, A, B \vdash \Delta} \RL{$\land_L$}
    \UICm{\Gamma, A \land B \vdash \Delta} \DP
    &
    \AXCm{\Gamma \vdash A, \Delta}
    \AXCm{\Gamma \vdash B, \Delta} \RL{$\land_R$}
    \BICm{\Gamma \vdash A \land B, \Delta} \DP
    \\\\

    \AXCm{\Gamma, A \vdash \Delta} \AXCm{\Gamma, B \vdash \Delta}
    \RL{$\lor_L$}
    \BICm{\Gamma, A \lor B \vdash \Delta}\DP
    &
    \AXCm{\Gamma \vdash A, B, \Delta} \RL{$\lor_R$}
    \UICm{\Gamma \vdash A \lor B, \Delta} \DP
    \\\\

    \AXCm{\Gamma \vdash A, \Delta} \AXCm{\Gamma, B \vdash \Delta}
    \RL{$\limp_L$}
    \BICm{\Gamma, A \limp B \vdash \Delta} \DP
    &

    \AXCm{\Gamma, A \vdash B, \Delta}
    \RL{$\limp_R$}
    \UICm{\Gamma \vdash A \limp B, \Delta} \DP
    \\\\

    \AXCm{\Gamma \vdash A, \Delta}
    \RL{$\neg_L$}
    \UICm{\Gamma, \neg A \vdash \Delta} \DP
    &

    \AXCm{\Gamma, A \vdash \Delta}
    \RL{$\neg_R$}
    \UICm{\Gamma \vdash \neg A, \Delta} \DP
    \\\\

    \AXCm{\Gamma, A[c/x] \vdash \Delta} \RL{ $\ex_L$}
    \UICm{\Gamma, \ex{x} A \vdash \Delta} \DP
    &
    \AXCm{\Gamma \vdash A[t/x], \Delta} \RL{
         $\ex_R$}
    \UICm{\Gamma \vdash \ex{x} A, \Delta} \DP

  \\\\
    \AXCm{\Gamma, A[t/x] \vdash \Delta} \RL{$\fa_L$}
    \UICm{\Gamma,  \fa{x} A \vdash \Delta} \DP
    &
    \AXCm{\Gamma \vdash A[c/x], \Delta} \RL{$\fa_R$}
    \UICm{\Gamma \vdash \fa{x} A, \Delta} \DP
    \\\\
    \AXCm{\Gamma, A, A \vdash \Delta} \RL{$\contr_L$}
    \UICm{\Gamma, A \vdash \Delta} \DP
    &
    \AXCm{\Gamma \vdash A, \Delta} \RL{$\contr_R$}
    \UICm{\Gamma \vdash A, A, \Delta} \DP
    \\\\
    \AXCm{\Gamma \vdash \Delta} \RL{$\weak_L$}
    \UICm{\Gamma, A \vdash \Delta} \DP
    &
    \AXCm{\Gamma \vdash \Delta} \RL{$\weak_R$}
    \UICm{\Gamma \vdash A, \Delta} \DP
    \\\\
  \end{tabular}

  where, in $\fa_L$ and $\ex_R$, $c$ is a fresh constant
  and, in $\fa_R$ and $\ex_L$, $t$ is any term.
\end{center}
}}
  \caption{Classical sequent calculus}
  \label{fig:clas}
\end{figure}

\begin{figure}
  \framebox[12.2cm][c]
{\parbox{12.2cm}
{\scriptsize
  \begin{center}
  \begin{tabular}{c@{\hspace{1cm}}c}
\multicolumn{2}{c}{   \AXCm{~} \RL{ax}
    \UICm{\Gamma, A \vdash A} \DP }
\\\\

    \AXCm{\Gamma, A, B \vdash \Delta} \RL{$\land_L$}
    \UICm{\Gamma, A \land B \vdash \Delta} \DP
    &
    \AXCm{\Gamma \vdash A}
    \AXCm{\Gamma \vdash B} \RL{$\land_R$}
    \BICm{\Gamma \vdash A \land B} \DP
    \\\\

    \AXCm{\Gamma, A \vdash \Delta} \AXCm{\Gamma, B \vdash \Delta}
    \RL{$\lor_L$}
    \BICm{\Gamma, A \lor B \vdash \Delta}\DP
    &
    \multicolumn{1}{c}{
    \AXCm{\Gamma \vdash A} \RL{$\lor_{R1}$~~~}
    \UICm{\Gamma \vdash A \lor B} \DP
    
    \AXCm{\Gamma \vdash B} \RL{$\lor_{R2}$}
    \UICm{\Gamma \vdash A \lor B} \DP
    }
    \\\\

    \AXCm{\Gamma \vdash A} \AXCm{\Gamma, B \vdash \Delta}
    \RL{$\limp_L$}
    \BICm{\Gamma, A \limp B \vdash \Delta} \DP
    &

    \AXCm{\Gamma, A \vdash B}
    \RL{$\limp_R$}
    \UICm{\Gamma \vdash A \limp B} \DP
    \\\\

    \AXCm{\Gamma \vdash A}
    \RL{$\neg_L$}
    \UICm{\Gamma, \neg A \vdash \Delta} \DP
    &

    \AXCm{\Gamma, A \vdash }
    \RL{$\neg_R$}
    \UICm{\Gamma \vdash \neg A} \DP
    \\\\

    \AXCm{\Gamma, A[c/x] \vdash \Delta} \RL{ $\ex_L$}
    \UICm{\Gamma, \ex{x} A \vdash \Delta} \DP
    &
    \AXCm{\Gamma \vdash A[t/x]} \RL{
         $\ex_R$}
    \UICm{\Gamma \vdash \ex{x} A} \DP

    \\\\
    \AXCm{\Gamma, A[t/x] \vdash \Delta} \RL{$\fa_L$}
    \UICm{\Gamma,  \fa{x} A \vdash \Delta} \DP
    &
    \AXCm{\Gamma \vdash A[c/x]} \RL{$\fa_R$}
    \UICm{\Gamma \vdash \fa{x} A} \DP
    \\\\
    \AXCm{\Gamma, A, A \vdash \Delta} \RL{$\contr_L$}
    \UICm{\Gamma, A \vdash \Delta} \DP
    &
    \\\\
    \AXCm{\phantom{A,} \Gamma \vdash \Delta} \RL{$\weak_L$}
    \UICm{\Gamma, A \vdash \Delta} \DP
    &
    \AXCm{\Gamma \vdash \phantom{A}} \RL{$\weak_R$}
    \UICm{\Gamma \vdash A} \DP
    \\\\
  \end{tabular}

  where, in $\fa_L$ and $\ex_R$, $c$ is a fresh constant
  and, in $\fa_R$ and $\ex_L$, $t$ is any term.
\end{center}
}}
  \caption{Intuitionistic sequent calculus}
  \label{fig:intuit}
\end{figure}

\subsection{Negative Translations} \label{sec:negative}

In this section, we briefly recall four existing translations \cite{FerOli10,ATroDvDal88}. In 1925, the first translation is published by Kolmogorov \cite{AKol25}. This translation involves adding a double negation in front of every subformula: 

\begin{tabular}{ rclrcl}
     $A^{Ko}$&$ \equiv$&$ \neg\neg A$  \textit{for A atomic} &
         $(\neg A)^{Ko} $&$ \equiv $&$ \neg\neg(\neg A^{Ko})$\\
     $(A \wedge B)^{Ko}$&$ \equiv$&$ \neg\neg(A^{Ko} \wedge B^{Ko})$&
         $(\forall x A)^{Ko} $&$ \equiv $&$ \neg\neg\forall x A^{Ko}$\\
     $(A \vee B)^{Ko}$&$ \equiv$&$ \neg\neg(A^{Ko} \vee B^{Ko})$&
         $(\exists x A)^{Ko}$&$  \equiv$&$  \neg\neg\exists x A^{Ko}$\\
     $(A \Rightarrow B)^{Ko} $&$\equiv$&$ \neg\neg(A^{Ko} \Rightarrow B^{Ko})$&
	& & \\
\end{tabular}\\

With the Kolmogorov's translation, $A$ is provable using classical
logic if and only if $A^{Ko}$ is provable 
using intuitionistic logic.\\

A few years later, G\"{o}del\cite{KGod33}\label{def:gg},
and independently Gentzen\cite{GGen36}, 
proposed a new translation, where disjunctions and existential
quantifiers are replaced by a combination of negation and their De Morgan
duals, respectively conjunctions and universal quantifiers: 

\begin{tabular}{r@{$~\equiv~$}l@{~~}r@{$~\equiv~$}l}
      $(A^{gg}) $&$ \neg\neg A$ \textit{for A atomic}&
         $(\neg A)^{gg}$&$ \neg A^{gg}$\\
     $(A \wedge B)^{gg}$&$ A^{gg} \wedge B^{gg}$&
        $(\forall x A)^{gg} $&$\forall x A^{gg}$  \\
     $(A \vee B)^{gg}$&$ \neg(\neg A^{gg}\wedge\neg B^{gg})$&
         $(\exists x A)^{gg}$&$\neg\forall x\neg A^{gg}$\\
    $(A \Rightarrow B)^{gg} $&$A^{gg}  \Rightarrow B^{gg}$ \\
\end{tabular}\\

As Kolmogorov's translation, G\"{o}del-Gentzen's translation allows to show that $A$ is provable using classical logic if and only if $A^{gg}$ is provable using intuitionistic logic.\\

Kuroda \cite{SKur51} defined in 1951 a new translation: 

\begin{tabular}{ r@{$~\equiv~$}l@{~~}rcl}
$A^{Ku}  $&$  A$ \textit{for A atomic} & $(\neg A)^{Ku}$&$ \equiv $&$ \neg A^{Ku}$\\
$(A \wedge B)^{Ku} $&$ A^{Ku} \wedge B^{Ku}$&$(\forall x A)^{Ku} $&$ \equiv $&$ \forall x \neg\neg A^{Ku}$ \\
$(A \vee B)^{Ku} $&$ A^{Ku} \vee B^{Ku}$& $(\exists x A)^{Ku}  $&$\equiv $&$ \exists x A^{Ku}$\\
$(A \Rightarrow B)^{Ku}  $&$A^{Ku} \Rightarrow B^{Ku}$
\end{tabular}\\

\noindent $A$ is provable classically if and only if $\neg \neg A^{Ku}$ is provable intuitionistically.\\

More recently, Krivine \cite{JKri90} has introduced a fourth translation: 

\begin{tabular}{r@{$~\equiv~$}lr@{$~\equiv~$}l}
$A^{Kr}$&$ \neg A$ \textit{for A atomic}& $(\neg A)^{Kr}$ & $\neg A^{Kr}$\\
$(A \wedge B)^{Kr} $&$ A^{Kr} \vee B^{Kr}$&$(\forall x A)^{Kr}$ & $\exists A^{Kr}$ \\
$(A \vee B)^{Kr} $&$ A^{Kr} \land B^{Kr}$&$(\exists x A)^{Kr}$ & $ \neg \exists x\neg A^{Kr}$\\
$(A \Rightarrow B)^{Kr}$&$\neg A^{Kr} \wedge B^{Kr}$
\end{tabular}\\

\noindent $A$ is provable classically if and only if $\neg A^{Kr}$ is provable intuitionistically.\\

Using these existing translations, in particular Kolmogorov's and G\"{o}del-Gentzen's translations, we propose to simplify them as described below.

\section{Polarizing Kolmogorov's Translation} \label{sec:Kpol}

As in Kolmogorov's translation, let us define the polarized Kolmogorov's translation:
\begin{definition} \label{def:polarization} 
\textit{Let A,B,C and D be propositions. An occurrence of A in B is}
\begin{itemize}    
  \item positive if:
	\begin{itemize}
 		\item B = A.
		\item B = $C \wedge D$ and the occurrence of A is in C or in D 
        and is positive.
		\item B = $C \vee D$ and the occurrence of A is in C or in D 
        and is positive.
		\item B = $C \Rightarrow D$  and the occurrence of A is in C (resp. in D) and is negative (resp. positive).
		\item B = $\neg C$ and the occurrence of A is in C and is negative.
		\item B = $\forall x C$  and the occurrence of A is in C and is positive.
		\item B = $\exists x C$  and the occurrence of A is in C and is positive.
	\end{itemize}
  \item negative if:
	\begin{itemize}
		\item B = $C \wedge D$ and the occurrence of A is in C or in D
        and is negative.
		\item B = $C \vee D$ and the occurrence of A is in C or in D
        and is negative.
		\item B = $C \Rightarrow D$  and the occurrence of A is in C (resp. in D)
        and is positive (resp. negative).
 		\item B = $\neg C$  and the occurrence of A is in C and is positive.
		\item B = $\forall x C$  and the occurrence of A is in C and is negative.
		\item B = $\exists x C$  and the occurrence of A is in C and is negative.
	\end{itemize}
\end{itemize}
\end{definition}

\begin{definition} \textit{Let A and B be propositions. We define by induction on the
 structure of propositions the positive ($K^+$) and negative translation ($K^-$): }

\begin{tabular}{ rclrcl}
      $ A^{K^+} $&$ \equiv $&$ A$ \textit{if A is atomic}&
          $ A^{K^-}  $&$\equiv$&$  A$ \textit{if A is atomic}\\
      $(A \wedge B)^{K^+} $&$\equiv$&$ A^{K^+} \wedge B^{K^+}$&
         $(A \wedge B)^{K^-} $&$\equiv$&$  \neg\neg A^{K^-} \wedge\neg\neg B^{K^-}$\\
     $(A \vee B)^{K^+} $&$\equiv$&$ A^{K^+} \vee B^{K^+}$&
        $(A \vee B)^{K^-} $&$\equiv$&$  \neg\neg A^{K^-} \vee \neg\neg B^{K^-}$\\
    $(A \Rightarrow B)^{K^+} $&$\equiv$&$ \neg\neg A^{K^-} \Rightarrow B^{K^+}$&
	$(A \Rightarrow B)^{K^-} $&$\equiv$&$ A^{K^+} \Rightarrow  \neg\neg B^{K^-}$ \\
    $(\neg A)^{K^+}  $&$\equiv$&$ \neg A^{K^-}$&
	$(\neg A)^{K^-}  $&$\equiv$&$\neg A^{K^+}$\\
     $(\forall x A)^{K^+}  $&$\equiv$&$ \forall x A^{K^+}$&
	$(\forall x A)^{K^-}  $&$\equiv$&$ \forall x \neg\neg A^{K^-}$ \\
    $(\exists x A)^{K^+}  $&$\equiv$&$ \exists x A^{K^+}$&
	$(\exists x A)^{K^-}  $&$\equiv$&$\exists x \neg\neg A^{K^-}$ \\
\end{tabular}
\end{definition}

Notice how, compared to \mysec{\ref{sec:negative}}, we introduce double negations in front of subformul\ae{} instead of the whole formula. For instance axioms are translated by themselves, and the price to pay is, as for Kuroda's and Krivine's translations, a negation of the whole formula in the following theorem.

\begin{theorem} \label{thm:pol-kol}
If the sequent $\Gamma \vdash \Delta$ is provable in the classical sequent
calculus then $\Gamma^{K^+}, \neg \Delta^{K^-} \vdash$ is
provable in the intuitionistic sequent calculus.
\end{theorem}
\begin{proof}
By induction on the proof-tree. Since this theorem is not the main
result of this paper, and is refined below (\mythm{\ref{thm:foc-int}}), let us process only one case. All other cases follow a similar pattern.

  \begin{tabular}{ccc}

    \AXCm{\pi}
    \UICm{\Gamma \vdash  \Delta, A[t/x]}
    \LeftLabel{$\exists_R$}
    \UICm{\Gamma \vdash \Delta,  \exists x A}
    \DP
    
    & $~~~\hookrightarrow~~$ &

    \AXC{IH($\pi$)}
    \UICm{\Gamma^{K^+},\neg\Delta^{K^-},\neg A^{K^-}[t/x]\vdash}
    \RightLabel{$\neg_R$}
    \UICm{\Gamma^{K^+},\neg \Delta^{K^-}\vdash \neg\neg A^{K^-}[t/x]}
    \RightLabel{$\exists_R$}
    \UICm{\Gamma^{K^+},\neg\Delta^{K^-}\vdash\exists x \neg\neg A^{K^-}}
    \RightLabel{$\neg_L$}
    \UICm{\Gamma^{K^+},\neg(\exists x \neg\neg A^{K^-}),\neg\Delta^{K^-}\vdash}
    \DP
  \end{tabular}
  
\noindent where IH($\pi$) denotes, here and later, the proof obtained by the application of the induction hypothesis on $\pi$.
\end{proof}

We also have the inverse translation.
\begin{theorem}
If the sequent $\Gamma^{K^+}, \neg \Delta^{K^-} \vdash   D^{K^-}$ is provable in the intuitionistic sequent calculus, then $\Gamma \vdash \Delta, D$  is provable in the classical sequent calculus.
\end{theorem}
\begin{proof}
By a straightforward induction on the proof-tree.
\end{proof}

We now focus on the polarization of the G\"{o}del-Gentzen's translation, which is lighter than the Kolmogorov's translation, again with the idea of getting a simpler translation in both directions.

\section{A Focused Sequent Calculus} \label{sec:foc}

G\"{o}del-Gentzen negative translation (\mydef{\ref{def:gg}} above)
removes many negations from translations and the polarization we
give in \mysec{\ref{sec:GGpol}} will even more. If we want to follow the pattern of \mythm{\ref{thm:pol-kol}} to show
equiprovability (in the absence of cut), we can {\em no longer}
systematically move formul\ae{} from the right to the left hand sides,
since we lack negation on almost all connectives. Therefore, we must
{\em constrain} our classical sequent calculus to forbid arbitrary
proofs, and in particular to impose that once a rule has been applied on
some formula of the right-hand side, the next rule must apply on the
corresponding subformula of the premiss. Working on the same formula
up to some well-chosen point is a discipline of capital importance, since we avoid to eagerly swap formul\ae{} from right to left.

This is why we introduce a focused version of the classical sequent
calculus. The resulting constraint is that we must decompose the {\em
  stoup} \cite{JYGir93,PLCurHHer00} formula until it gets removed from the {\em
  stoup} position. Only when the stoup becomes empty, can we apply rules on other formul\ae{}.

\begin{definition}[Focused sequent]
A focused sequent is a triple, composed of two multisets of
formul\ae{} and a distinguished set (the {\em stoup}) containing zero
or one formula. It will be noted $\Gamma \vdash A; \Delta$ when the
distinguished set contains a formula $A$, and $\Gamma \vdash .; \Delta$
when it contains no formula.
\end{definition}

The focused sequent calculus we define serves our particular purpose;
for instance it is not optimized to maximize the so-called negative and positive phases
\cite{CLiaDMil09}. Note also that in our paper, negative and positive
has a very different meaning. The calculus is presented in
\myfig{\ref{fig:focus}} and contains a stoup only in the right-hand
side, since this is the only problematic side.

Note that all the left rules require an empty stoup, and that two new right rules, \texttt{focus} and \texttt{release}, respectively place and remove a formula of the right-hand side in the focus.

Only atomic, negated, disjunctive or existentially quantified formul\ae{} can be removed from the stoup:
\begin{itemize}
\item Due to the freshness condition of the $\ex$-left and $\fa$-right rule, the $\ex$-right rule is the only rule that cannot be inverted (or equivalently permuted downwards). Therefore existential statements must be removable from the stoup.

\item The stoup has only one place, so we cannot allow in it both subformul\ae{} of a disjunction. This choice must be done by a subsequent call to the focus rule. More pragmatically, G\"{o}del-Gentzen's translation introduces negations in this case, enabling the storage of the subformul\ae{} on the left-hand side of the sequent. As an informal translation rule, intuitionistic $\neg_R$ rules will correspond to a lost of focus.
\item The same reasoning holds for allowing atomic formul\ae{} to be
  removed from the stoup. Also, if we do not allow this, the system loses completeness since the stoup becomes stuck forever. 
\item Allowing to remove negated formul\ae{} from the stoup
  accounts for the aggressive behavior of the operator $\gen$: to keep the statement of \mythm{\ref{thm:foc-int}} short and close to statements of previous theorems, we must remember that $\gen$ removes the negation of negated formul\ae{}, therefore forcing them to move on the left hand side.
\end{itemize}

As a consequence of the design constraint imposed by our translation, the rule focus cannot act on a
formula which has $\ex$, $\neg$ or $\lor$ as main connective and the
$\ex_R$, $\neg_R$ and $\lor_R$ rules act on formul\ae{} that are not
in the stoup (and, as mentioned, when the stoup itself is
empty). The reasons become clear in the proof of \mythm{\ref{thm:foc-int}}.

Lastly, we impose the formula in the axiom rule to be atomic,
which boils down to an $\eta$-expansion of the usual axiom rule.

To sum up, we consider the connectives $\ex$, $\lor$ and $\neg$, when
they appear on the right-hand side of a sequent, to have a "positive phase" in the sense of \cite{CLiaDMil09} and the other ones to have a negative phase.

\begin{figure}[ht]
  \framebox[12.2cm][c]
{\parbox{12.2cm}
{\scriptsize
  \begin{center}
  \begin{tabular}{c@{\hspace{1cm}}c}
\multicolumn{2}{c}{   \AXCm{~} \RL{ax}
    \UICm{\Gamma, A \vdash . 
    ; A, \Delta} \DP }
\\\\

    \AXCm{\Gamma, A, B \vdash .; \Delta} \RL{$\land_L$}
    \UICm{\Gamma, A \land B \vdash .; \Delta} \DP
    &
    \AXCm{\Gamma \vdash A; \Delta}
    \AXCm{\Gamma \vdash B; \Delta} \RL{$\land_R$}
    \BICm{\Gamma \vdash A \land B; \Delta} \DP
    \\\\

    \AXCm{\Gamma, A \vdash .; \Delta} \AXCm{\Gamma, B \vdash .; \Delta}
    \RL{$\lor_L$}
    \BICm{\Gamma, A \lor B \vdash .; \Delta}\DP
    &
    \AXCm{\Gamma \vdash .; A, B, \Delta} \RL{$\lor_R$}
    \UICm{\Gamma \vdash .; A \lor B, \Delta} \DP
    \\\\

    \AXCm{\Gamma \vdash A; \Delta} \AXCm{\Gamma, B \vdash .; \Delta}
    \RL{$\limp_L$}
    \BICm{\Gamma, A \limp B \vdash .; \Delta} \DP
    &

    \AXCm{\Gamma, A \vdash B; \Delta}
    \RL{$\limp_R$}
    \UICm{\Gamma \vdash A \limp B; \Delta} \DP
    \\\\

    \AXCm{\Gamma \vdash A; \Delta}
    \RL{$\neg_L$}
    \UICm{\Gamma, \neg A \vdash .; \Delta} \DP
    &

    \AXCm{\Gamma, A \vdash .; \Delta}
    \RL{$\neg_R$}
    \UICm{\Gamma \vdash .; \neg A, \Delta} \DP
    \\\\

    \AXCm{\Gamma, A[c/x] \vdash .; \Delta} \RL{ $\ex_L$}
    \UICm{\Gamma, \ex{x} A \vdash .; \Delta} \DP
    &
    \AXCm{\Gamma \vdash .; A[t/x], \Delta} \RL{
         $\ex_R$}
    \UICm{\Gamma \vdash .; \ex{x} A, \Delta} \DP
  \\\\
    \AXCm{\Gamma, A[t/x] \vdash .; \Delta} \RL{$\fa_L$}
    \UICm{\Gamma,  \fa{x} A \vdash .; \Delta} \DP
    &
    \AXCm{\Gamma \vdash A[c/x]; \Delta} \RL{$\fa_R$}
    \UICm{\Gamma \vdash \fa{x} A; \Delta} \DP
    \\\\
    \AXCm{\Gamma, A, A \vdash .; \Delta} \RL{$\contr_L$}
    \UICm{\Gamma, A \vdash .; \Delta} \DP
    &
    \AXCm{\Gamma \vdash .; A, \Delta} \RL{$\contr_R$}
    \UICm{\Gamma \vdash .; A, A, \Delta} \DP
    \\\\
    \AXCm{\Gamma \vdash .; \Delta} \RL{$\weak_L$}
    \UICm{\Gamma, A \vdash .; \Delta} \DP
    &
    \AXCm{\Gamma \vdash .; \Delta} \RL{$\weak_R$}
    \UICm{\Gamma \vdash A; \Delta} \DP
    \\\\
    & 
    \AXCm{\Gamma \vdash A; \Delta} \RL{focus}
    \UICm{\Gamma \vdash .; A, \Delta} \DP
    \\\\
    &
    \AXCm{\Gamma \vdash .; A, \Delta} \RL{release}
    \UICm{\Gamma \vdash A; \Delta} \DP
    \\\\
  \end{tabular}\\
\end{center}
where:
\begin{itemize}
\item the axiom rule involves only atomic formul\ae{},
\item in $\fa_L$ and $\ex_R$, $c$ is a fresh constant, 
\item in $\fa_R$ and $\ex_L$, $t$ is any term, 
\item in release, $A$ is either atomic or of the form $\ex{x} B, B \lor C$ or $\neg B$,
\item in focus, $A$ is neither atomic nor of the form  $\ex{x} B, B \lor
  C$ or $\neg B$.
\end{itemize}

}}
  \caption{Focused classical sequent calculus}
  \label{fig:focus}
\end{figure}

We show that this calculus is equivalent to the usual sequent
calculus of \myfig{\ref{fig:clas}}.

\begin{proposition}
Let $\Gamma, \Delta$ be two multisets of formul\ae{} and $A$ be a
formula. If the sequent $\Gamma \vdash .; \Delta$ (resp. $\Gamma
\vdash A; \Delta$) has a proof in the focused sequent calculus, then
it has a proof in the classical sequent calculus. 
\end{proposition}
\begin{proof}
Straightforward by noticing that, forgetting about the stoup
(transforming the semicolon into a comma), all focused rules are
instances of the classical sequent calculus rules. Both rules {\em focus}
and {\em release} lose their meaning and are simply erased from the
proof-tree. \qed{}
\end{proof}

The converse is a corollary of the slightly more general following statement. As we see below, it is crucial to have some degree of freedom to decompose arbitrarily $\Delta'$ into $A$ and $\Delta$ in order to reason properly by induction.

\begin{proposition}\label{prop:clas-focus-equiv}
Let $\Gamma, \Delta'$ be two multisets of formul\ae{}. Assume that the
sequent $\Gamma \vdash \Delta'$ has a proof in the classical sequent
calculus. Let $A$ be a set containing either a formula (also named
$A$ by abuse of notation) or the empty formula, and let $\Delta$ such
that $\Delta' =  A, \Delta$.\\

 Then the sequent $\Gamma \vdash A; \Delta$ has a proof in the focused
 sequent calculus.
\end{proposition}
\begin{proof}
The proof is a little bit more involved, but it appeals only to simple
and well-known principles, in particular to Kleene's inversion lemmas \cite{SKle52,OHer10}, stating that inferences rules can be permuted and, therefore, gathered. 

We give only a sketch of the proof, leaving out the details to the reader, for two reasons. Firstly, giving all the lengthy details would not add any insight on the structure of the proof; in the contrary they would blur the visibility of the main ideas. Secondly, similar completeness results are known for much more constrained focused proof systems; see for instance the one presented in \cite{CLiaDMil09}.\\

First of all, we consider a refined version of the classical sequent calculus of \myfig{\ref{fig:clas}} where proofs are restricted to use the axiom and weak rules on atomic formul\ae{}. In this way, we know \cite{OHer10} that Kleene's inversion lemmas \cite{SKle52} make the proof height {\em decrease} strictly. We reason by induction of the {\em height} of this modified proof-tree $\pi$, distinguishing the three following cases: 
\begin{itemize}
\item $A$ is empty, or $A$ contains an atomic, existential,
  disjunctive or negated formula that is not the active formula of the
  last rule $r$ of $\pi$. Then we release $A$,
  focus on the active formula if necessary, apply rule $r$, and we
  get one or two premises, on which we can apply the induction
  hypothesis. Let us give two instances, where, in the second case, $A$
  is empty:
  
  \begin{center}
  \begin{tabular}{ccc}
    \AXCm{\pi}
    \UICm{\Gamma \vdash B, A, \Delta}
    \LeftLabel{$\neg_L$}
    \UICm{\Gamma, \neg B \vdash A, \Delta}
    \DP
    & $\hookrightarrow$ &
    \AXC{IH($\pi$)}
    \UICm{\Gamma \vdash B; A, \Delta}
    \RightLabel{$\neg_L$}
    \UICm{\Gamma, \neg B \vdash .; A, \Delta}
    \RightLabel{release}
    \UICm{\Gamma, \neg B \vdash A; \Delta}
    \DP
    \\ 
    &
    \\
    \AXCm{\pi_1}
    \UICm{\Gamma \vdash B, \Delta}
    \AXCm{\pi_2}
    \UICm{\Gamma \vdash C, \Delta}
    \LeftLabel{$\land_R$}
    \BICm{\Gamma \vdash B \land C, \Delta}
    \DP
    & $~~\hookrightarrow~~$ &
    \AXC{IH($\pi_1$)}
    \UICm{\Gamma \vdash B; \Delta}
    \AXC{IH($\pi_2$)}
    \UICm{\Gamma \vdash C; \Delta}
    \RightLabel{$\land_R$}
    \LeftLabel{~~}
    \BICm{\Gamma \vdash B \land C; \Delta}
    \RightLabel{focus}
    \UICm{\Gamma \vdash .; B \land C, \Delta}
    \DP
  \end{tabular}
\end{center}

\item $A$ contains an atomic, existential, disjunctive or negated
  formula that is active in the last rule $r$ of $\pi$. Then $r$ must
  be one of the six rules axiom, $\ex_R$, $\lor_R$,
  $\neg_R$, weak${}_R$ or contr${}_R$. They are direct and all the
remaining cases follow a similar pattern. Here is the case for the $\ex_R$ rule:

\begin{center}
  \begin{tabular}{ccc}
    \AXCm{\pi}
    \UICm{\Gamma \vdash B[t/x], \Delta}
    \LeftLabel{$\ex_R$}
    \UICm{\Gamma \vdash \ex{x} B, \Delta}
    \DP
    & $~~\hookrightarrow~~$ &
    \AXC{IH($\pi$)}
    \UICm{\Gamma \vdash . ; B[t/x], \Delta}
    \RightLabel{$\ex_R$}
    \UICm{\Gamma \vdash .; \ex{x} B, \Delta}
    \RightLabel{release}
    \UICm{\Gamma \vdash  \ex{x} B ; \Delta}
    \DP
    \\ 
  \end{tabular}
\end{center}

\item If $A$ is not empty and not an atomic, existential, disjunctive
  or negated formula then, disregarding the last rule of
  $\pi$, we apply Kleene's inversion lemma on $A$, the induction
  hypothesis on the premises, since the proof height has decreased,
  and recompose those premises to get back the corresponding component(s) of $A$ in the stoup. Here is an example of such a rule : 

\begin{center}
  \begin{tabular}{ccc}
    \AXCm{\pi}
    \UICm{\Gamma, B \vdash C, \Delta}
    \LeftLabel{$\limp_R$}
    \UICm{\Gamma \vdash B \limp C, \Delta}
    \DP
    & $\hookrightarrow$ &
    \AXC{IH($\pi$)}
    \UICm{\Gamma, B \vdash C; \Delta}
    \RightLabel{$\limp_R$}
    \UICm{\Gamma \vdash B \limp C; \Delta}
    \DP
    \\ 
\end{tabular}
\end{center}

\end{itemize}
\qed{}
\end{proof}

\section{Polarizing G\"odel-Gentzen's Translation} \label{sec:GGpol}

We try to reduce the number of negations. We use the polarization 
of propositions (\mydef{\ref{def:polarization}} above) and replace 
disjunction and existential quantifiers by conjunction and universal quantifiers, as in G\"odel-Gentzen's translation.

\begin{definition} \label{def:pol-god}
Let A and B be propositions. We define, by induction on the structure 
of propositions, the positive(${}^p\!\!$) and negative(${}^n\!\!$) translations: 

\begin{tabular}{ rclrcl}
      $A^p $&$\equiv$&$ A$ \textit{if A is atomic}&
          $ A^n  $&$\equiv$&$ \neg\neg A$ \textit{if A is atomic}\\
     $(A \wedge B)^p $&$\equiv$&$ A^p \wedge B^p$&
         $(A \wedge B)^n $&$\equiv$&$   A^n \wedge B^n$\\
     $(A \vee B)^p $&$\equiv$&$  A^p \vee  B^p$&
        $(A \vee B)^n $&$\equiv$&$  \neg(\neg A^n \wedge \neg B^n)$\\
    $(A \Rightarrow B)^p $&$\equiv$&$ A^n \Rightarrow B^p$&
	$(A \Rightarrow B)^n $&$\equiv$&$ A^p \Rightarrow  B^n$ \\
    $(\neg A)^p  $&$\equiv$&$ \neg A^n$&
	$(\neg A)^n  $&$\equiv$&$\neg A^p$\\
     $(\forall x A)^p  $&$\equiv$&$ \forall x A^p$&
	$(\forall x A)^n  $&$\equiv$&$ \forall x  A^n$\\
    $(\exists x A)^p  $&$\equiv$&$ \exists x  A^p$&
	$(\exists x A)^n  $&$\equiv$&$ \neg\forall x \neg A^n$\\
\end{tabular}
\end{definition}

\begin{theorem}\label{thm:foc-int}
Let $\Gamma, \Delta$ be multisets of formul\ae{}, and $A$ be a set
containing zero or one formula. If the sequent $\Gamma \vdash A;
\Delta$ has a proof in the (classical) focused sequent calculus, then,
in the intuitionistic sequent calculus, the sequent $\gpos{\Gamma},
\gen \gneg{\Delta} \vdash \gneg{A}$ has a proof.

\end{theorem}

Notice that $\gen$ removes the trailing negation of $\Delta$ in three
cases: the negative translations of $\ex$, of $\lor$ and of $\neg$
(this last one more as a side-effect).

\begin{proof}
By induction on the proof of $\Gamma \vdash A; \Delta$, considering
one by one the 19 cases from \myfig{\ref{fig:focus}}:
\begin{itemize}
\item A left rule. We apply the induction hypothesis on the premises and
  copy the left rule. For instance, if the rule is $\limp_L$, then the
  induction hypothesis gives us proofs of the two sequents
  $\gpos{\Gamma}, \gen \gneg{\Delta} \vdash \gneg{A}$ (since $A$ is put in the stoup in the $\limp_L$ rule) and
  $\gpos{\Gamma}, \gpos{B}, \gen \gneg{\Delta} \vdash$ that can be
  readily combined with the (intuitionistic) $\limp_L$ rule to yield a
  proof of the sequent $\gpos{\Gamma}, \gneg{A} \limp \gpos{B}, \gen
  \gneg{\Delta} \vdash$. This is what we were looking for, since
  $\gpos{(A \limp B)} \equiv \gneg{A} \limp \gpos{B}$.

\item A contr${}_R$ rule. It is transformed (after application of the
  induction hypothesis) into a contr${}_L$ rule.

\item A weak${}_R$ rule. It is transformed into a weak${}_R$ rule.

\item A release rule. This can occur only if $A$ is atomic or of the
  form $\ex{x} B$, $B \lor C$ or $\neg B$. In all cases, we translate
  it as a $\neg_R$ rule, which removes the trailing negation of
  $\gneg{A}$, turning it into the formula $\gen \gneg{A}$ (see \mydef{\ref{def:antinegation}}), that integrates directly $\gen \gneg{\Delta}$, so that we can readily plug the
  proof obtained by the application of the induction hypothesis.

\item A focus rule on $A \in \Delta$. This can occur only if $A$ is
  neither atomic nor of the form $\ex{x} B$, $B \lor C$ or $\neg
  B$. Therefore, $\gen \gneg{A} = \neg \gneg{A}$, and we apply a
  $\neg_L$ rule.

\item An axiom rule. Since $A$ is restricted to be atomic, we
  need to build an intuitionistic proof of the sequent $\gpos{\Gamma},
  A, \gen \neg \neg{A}, \gen \gneg{\Delta}\vdash $ which is a trivial two-step proof since $\gen \neg \neg {A}$ is $\neg A$.

\item A $\neg_R$ rule:
\begin{prooftree}
\AXCm{\pi}
\UICm{\Gamma, B \vdash .; \Delta}
\UICm{\Gamma \vdash .; \neg B, \Delta}
\end{prooftree}

We must find a proof of the sequent $\gpos{\Gamma}, \gen \neg
\gpos{B}, \gen \gneg{\Delta} \vdash$. But $\gen \neg
\gpos{B}$ {\em is the same as} $\gpos{B}$, and the induction hypothesis
gives us directly a proof of $\gpos{\Gamma}, \gpos{B}, \gen \gneg{\Delta}
\vdash$. In other words, $\neg_R$ is not translated, thanks to the
operator $\gen$ (that will soon lead us into minor considerations).

\item A $\lor_R$ rule:
\begin{prooftree}
\AXCm{\pi_1}
\UICm{\Gamma \vdash .; B, C, \Delta}
\UICm{\Gamma \vdash .; B \lor C, \Delta}
\end{prooftree}

We must build a proof of the sequent $\gpos{\Gamma}, \gen \neg
(\neg \gneg{B} \land \neg \gneg{C}), \gen \gneg{\Delta} \vdash$, which
is equal to $\gpos{\Gamma}, \neg \gneg{B} \land \neg \gneg{C}, \gen \gneg{\Delta} \vdash$. It is natural to try to apply the $\land_L$ rule:

\begin{prooftree}
   \AXCm{\gpos{\Gamma}, \neg \gneg{B}, \neg \gneg{C}, \gen \gneg{\Delta} 
  	\vdash}
  \RightLabel{$\land_L$}
  \UICm{\gpos{\Gamma}, \neg \gneg{B} \land \neg \gneg{C}, \gen \gneg{\Delta} 
  	\vdash}
\end{prooftree}

We are committed to find a proof of the premiss, while the induction hypothesis gives us a proof of the following slightly different sequent:

$$
\gpos{\Gamma}, \gen \gneg{B}, \gen \gneg{C}, \gen \gneg{\Delta} \vdash
$$ 

Therefore we must examine two subcases:

\begin{itemize}
  \item $B$ is an atom, an existential, disjunctive or negated
    formula. Then $\gneg{B} = \neg D$ for some $D$, and $\gen \gneg{B}
    = D$. We build the following proof, given the proof obtained by
    application of the induction hypothesis:
    \begin{prooftree}
      \AXCm{\gpos{\Gamma}, D, \gen \gneg{C}, \gen \gneg{\Delta} \vdash}
      \RightLabel{$\neg_R$}
      \UICm{\gpos{\Gamma}, \gen \gneg{C}, \gen \gneg{\Delta} \vdash \neg D}
      \RightLabel{$\neg_L$}
      \UICm{\gpos{\Gamma}, \neg \gneg{B}, \gen \gneg{C}, \gen \gneg{\Delta} \vdash}
    \end{prooftree}
  \item otherwise $\gen \gneg{B} = \neg \gneg{B}$, and the induction hypothesis gives us directly a proof of the above sequent.
\end{itemize}
We do a similar case distinction on $C$ to get from the previous proof a proof of the sequent $\gpos{\Gamma}, \neg \gneg{B}, \neg \gneg{C}, \neg \gneg{\Delta} \vdash$, which is now exactly what we were looking for.

\item A $\ex_R$ rule:
  \begin{prooftree}
    \AXCm{\pi}
    \UICm{\Gamma \vdash .; A[t/x], \Delta}
    \UICm{\Gamma \vdash .; \ex{x} A, \Delta}
  \end{prooftree}

The induction hypothesis gives us a proof of the sequent $\gpos{\Gamma},
\gen \gneg{A[t/x]}, \gen \gneg{\Delta} \vdash$, that we turn, in the
same way as in the previous case, into a proof of the sequent
$\gpos{\Gamma}, \neg \gneg{A[t/x]}, \gen \gneg{\Delta} \vdash$, to which we apply the $\fa_L$ rule:
\begin{prooftree}
\AXCm{\gpos{\Gamma}, \neg \gneg{A[t/x]}, \gen \gneg{\Delta} \vdash}
\LeftLabel{$\fa_L$}
\UICm{\gpos{\Gamma}, \fa{x} \neg \gneg{A}, \gen \gneg{\Delta} \vdash}
\end{prooftree}

the end sequent is also equal to $\gpos{\Gamma}, \gen \gneg{\ex{x}
  A}, \gen \gneg{\Delta} \vdash$; so we have exhibited the proof we
were looking for.

\item A $\land_R$, $\limp_R$ or $\fa_R$ rule. Those three last cases are easy, since we are in the stoup, which corresponds to the right-hand side of the (intuitionistic) sequent. For example, let us consider the case of the $\limp_R$ rule:

\begin{center}
  \begin{tabular}{ccc}
    \AXCm{\phantom{I(}\pi\phantom{)I}}
    \UICm{\Gamma, A \vdash B; \Delta}
    \LeftLabel{$\limp_R$}
    \UICm{\Gamma \vdash A \limp B; \Delta}
    \DP
    & $~~\hookrightarrow~~$ &
    \AXC{IH($\pi$)}
    \UICm{\gpos{\Gamma}, \gpos{A}, \gen \gneg{\Delta} \vdash \gneg{B}}
    \RightLabel{$\limp_R$}
    \UICm{\gpos{\Gamma}, \gen \gneg{\Delta} \vdash \gpos{A} \limp \gneg{B}}
    \DP
    \\ 
\end{tabular}
\end{center}
\end{itemize}
\qed{}
\end{proof}

The reverse translation is expressed with respect to the unfocused sequent calculus, which is more liberal and therefore more convenient for the reverse way. We nevertheless need to slightly generalize the statement.

\begin{theorem}\label{thm:foc-inv}
Let $\Gamma, \Delta_1, \Delta_2$ be multisets of formul\ae{}, such that $\Delta_1$ does not contain any negated formula. Let $D$ be at most one formula.
If the sequent $\Gamma^p, \gen \Delta_1^n, \neg \Delta_2^n \vdash  D^n$ is provable in the intuitionistic sequent calculus, then $\Gamma \vdash \Delta_1, \Delta_2, D$ is provable in the classical sequent calculus.
\end{theorem}

\begin{proof}
We constraint the intuitionistic proof to have a certain shape before starting the induction. First, we assume that axiom rules are restricted to atoms. It is always possible to expanse the axioms that are not of this form. Second, we also assume that $\neg_L$ rules on atomic forml\ae{} are permuted upwards as far as they can \cite{SKle52}, this basically induces that this $\neg_L$ rule becomes glued either to an axiom or to a weakening rule and therefore the axiom case will be integrated to $\neg_L$ case. This way, we avoid the presence of non-double negated axioms on the right-hand side. Unless stated otherwise we do not mention explicitly the application of the induction hypothesis, which is clear from the context.
\begin{itemize}
\item A contraction or a weakening on any of the formul\ae{} of $\Gamma^p, \gen \Delta_1^n, \neg \Delta_2^n$ or $D^n$ is turned into the same rule on the corresponding formula of $\Gamma, \Delta_1, \Delta_2$ or $D$. Below, we now concentrate on connective and quantifier rules.

\item A left-rule on $\Gamma^p$ is turned into the same left-rule on $\Gamma$. The potential erasing of $D^n$ in the two cases $\limp_L$ and $\neg_L$ is handled through a weakening.

\item A left-rule on $\neg \Delta_2^n$ can be only a $\neg_L$ rule. It is turned into a weakening on $D^n$ if necessary, since we apply the induction hypothesis on the premiss $\Gamma^p, \gen \Delta_1^n, \neg {(\Delta'_2)}^n \vdash D_2^n$, with $\Delta_2 = \Delta'_2, D_2$.

\item A right-rule on $D^n$, assuming the main connective or quantifier of $D$ is $\land$, $\limp$ or $\fa$. The rule is $\land_R$, $\limp_R$ or $\fa_R$, respectively. It is turned into the same right-rule on $D$.

\item A right-rule on $D^n$, assuming $D$ is an atomic, existentially quantified or disjunctive formula. The rule is $\neg_R$ and the premiss is $\Gamma^p, \gen {\Delta_1}^n, \gen D^n, \neg \Delta_2^n \vdash$, to which we only need to apply the induction hypothesis.

\item A right-rule on $D^n$ assuming $D$ is a negated formula $\neg D'$. In this case, the rule must be $\neg_R$ and the premiss is of the form $\Gamma^p, \gen {\Delta_1}^n, \gen (D)^n, \neg \Delta_2^n \vdash$. But $D$ is negated, and $\gen D^n = \gen \neg D'^p = D'^p$. So, to apply the induction hypothesis, we consider that the premiss is $\Gamma^p, D'^p, \gen {\Delta_1}^n, \neg \Delta_2^n \vdash$.

\item All possibilities for $D$ have been examined. Notice in particular that a $\lor_R$ or $\ex_R$ rule cannot be applied on $D^n$, since the negative translation of \mydef{\ref{def:pol-god}} never introduces this connective (resp. quantifier) in head position.

\item A left-rule on $\gen D_1^n \in \gen \Delta_1^n$, assuming $D_1$ is a disjunctive (resp. existentially quantified) formula. The rule is $\land_L$ (resp. $\fa_L$) and is turned in an $\lor_R$ (resp. $\ex_R$) rule on $D_1$, making the active formula(e) of the premiss(es) move from $\Delta_1$ to $\Delta_2$ if necessary. Let us detail the $\lor$ case. $D_1 = B_1 \lor C_1$, and $\gen D_1^n = \neg B_1^n \land \neg C_1^n$. The $\land_L$ rule gives us the premiss :
$$
\Gamma_1^p, \gen (\Delta'_1)^n, \neg B_1^n, \neg C_1^n, \neg \Delta_2^n \vdash D^n
$$
We must distinguish according to the shapes of $B_1$ and $C_1$. Let us discuss only $B_1$, the discussion on $C_1$ being exactly the same:
	\begin{itemize}
    	\item $B_1$ is an atomic, existentially quantified, disjunctive or negated formula. Then $\neg B_1^n$ is different from $\gen B_1^n$, and to apply properly the induction hypothesis, $B_1$ must be placed into $\Delta_2$.
        \item Otherwise $\neg B_1^n = \gen B_1^n$ and it is left into $\Delta_1$.
    \end{itemize}

\item A left-rule on $\gen D_1^n \in \gen \Delta_1^n$, assuming that the main connective or quantifier of $D_1$ is $\limp$, $\land$ or $\fa$. The rule is $\neg_L$ in all cases, and we only need to weaken on $D^n$ if necessary, before applying the induction hypothesis on the premises. 

\item A left rule on $\gen D_1^n \in \gen \Delta_1^n$, assuming $D_1$ is atomic. By assumption, the next rule is a rule on $D_1$. If it is a weakening, we translate both rules at once by weakening on $D_1$ and apply the induction hypothesis to the premiss of the weakening rule. Otherwise the next rule is an axiom. The only possibility is that $D_1$ belongs to $\Gamma^p$, and we translate both rules as an axiom.

\item By assumption, $D_1$ cannot be a negated formula, and therefore all the cases have been considered. \qed{}
\end{itemize}
\end{proof}

\begin{corollary}
Let $\Gamma, \Delta$ be multisets of formul\ae{} and $D$ be at most one formula. If the sequent $\Gamma^p, \gen \Delta^n \vdash D^n$ is provable in the intuitionistic sequent calculus, then the sequent $\Gamma \vdash \Delta, D$ is provable in the classical sequent calculus.
\end{corollary}
\begin{proof}
Let $\neg C_1, \cdots, \neg C_n$ be the negated formul\ae{} of $\Delta$ and $\Delta'$ the other ones. We apply \mythm{\ref{thm:foc-inv}} to a $\Gamma$ composed of $\Gamma, C_1, \cdots, C_n$, a $\Delta_1$ composed of $\Delta'$, an empty $\Delta_2$ and finally a $D$ equal to $D$, which gives a proof of the sequent:
$$
\Gamma, C_1, \cdots, C_n \vdash \Delta', D
$$
to which we apply $n$ times the $\neg_R$ rule to get back a proof of the wanted sequent.
\end{proof}

\section{Conclusion and Further Work} \label{sec:conc}

In this paper, we have shown that polarized double-negation translations still are used to navigate between intuitionistic and classical logics. They are lighter in terms of double negation, and let more statements being invariant by translation.

For instance, consider the axiom $(A \land B) \limp (A \lor B)$. Kolmogorov's translation introduces 14 negations: $\neg\neg (\neg\neg (\neg\neg A \land \neg\neg B) \limp \neg\neg (\neg\neg A \lor \neg\neg B))$, while its (positive) polarized variant, only 10 of them: $(\neg\neg (\neg\neg A \land \neg\neg B) \limp (\neg\neg A \lor \neg\neg B)$. G{\"o}del-Gentzen's translation would be $(\neg\neg A \land \neg\neg B) \limp \neg (\neg \neg \neg A \land \neg \neg \neg B)$, introducing 11 negations, while its polarized version introduces only 4 of them: $(\neg \neg A \land \neg \neg B) \limp (A \lor B)$. Recent work from Fr\'ederic Gilbert tends to show that it is possible to go further in the removal of double-negations by turning double negations into an operator that analyses the structure of the formula it double-negates. The polarization of Krivine's translation remains also to be examined.\\

Polarized translations are particularity adapted to cut-free proofs; otherwise the same active formula may appear both in the left and the right hand sides. As the negative and positive translations of a formula usually differ, it is impossible to cut them back. The workaround can be a ``manual'' elimination of this cut by reductive methods up to the point where both translations become equal, or to loosen the intuitionistic cut rule. We can also decide not to bother with cuts by eliminating them a priori. In all cases, however, we rely on a cut-elimination theorem that does not hold in the general case of the application described below.\\

Polarized double-negation translations has been primarily designed to fit {\em polarized deduction modulo} \cite{GDow10}, an extension of first-order logic by a congruence on formul\ae{} that is generated by {\em polarized} rewrite rules that apply only on a given side of the turn-style. It has already led to interesting results \cite{GBur10,GBur11} in automated theorem proving within axiomatic theories. To support this approach, we must ensure the cut-elimination property of the (sequent calculus modulo the polarized) rewrite system.

One canonical way is to first show proof normalization for the natural deduction, and shift this result to the intuitionistic sequent calculus. Then, through a double-negation translation of {\em the rewrite system} this result can be extended to the classical sequent calculus \cite{GDowBWer03}. In case of polarized rewriting, a polarized translation can be of great help for this last step, in addition to the development of normalization proofs via reducibility candidates. Another way to get cut admissibility would be to develop semantic proofs.

Lastly, it could be interesting to investigate whether, even in absence of cut admissibility as it can be the case, the modularity of our translations can be enforced, or whether cuts between two differently translated left- and right-formul\ae{} can nevertheless be eliminated. We conjecture that this is possible, provided the rewrite relation is confluent and terminating.\\

\noindent {\em Acknowledgments.} The authors would like to thank P. Jouvelot for his comments and the choice of the title of this paper.

\bibliographystyle{splncs}
\bibliography{betternot}

\end{document}